\documentclass[final,twoside,11pt]{entics} 
\usepackage{enticsmacro}
\usepackage{graphicx}
\usepackage{proof}
\usepackage{framed}
\usepackage{amssymb}
\usepackage{mathtools}
\newcommand{\IF}{\mbox{{\bf if}\ }}
\newcommand{\FI}{\mbox{\ {\bf fi}}}
\newcommand{\DO}{\mbox{\ {\bf do}\ }}
\newcommand{\OD}{\mbox{\ {\bf od}}}

\newcommand{\WHILE}{\mbox{{\bf while}\ }}
\newcommand{\THEN}{\mbox{\ {\bf then}\ }}
\newcommand{\ELSE}{\mbox{\ {\bf else}\ }}

\newcommand{\HT}[3]{\ensuremath{\{{#1}\}\ {#2}\ \{{#3}\}}}
\newcommand{\T}{\mbox{{\bf true}}}
\newcommand{\F}{\mbox{{\bf false}}}

\newcommand{\new}{\mbox{{\bf cons}}}
\newcommand{\delete}{\mbox{{\bf dispose}}}

\newcommand{\update}[1]{ \langle  #1 \rangle }
\newcommand{\dispose}[1]{\langle #1 \rangle:=\bot}

\newcommand{\sep}{\mathrel{*}}
\newcommand{\sepimp}{\mathrel{-\kern-.2em *}}
\newcommand{\nothookrightarrow}{\kern.2em\not\kern-.2em\hookrightarrow}

\def\conf{MFPS 2023} 	%%Fill in the acronym for your conference (with year)
\def\myconf{mfps39}
\volume{3}			%Fill in the ENTICS volume number here
\def\volu{3}			% and here
\def\papno{5}				%Fill in your paper number here
\def\lastname{De Boer, Hiep, De Gouw} %Lastnames appear in the running header 
\def\runtitle{Dynamic Separation Logic} %Short title appears in the running header on even pages. 
\def\mydoi{12297}% {{\normalshape  \href{https://doi.org/10.46298/entics.10354}{10.46298/entics.10354}}}
\def\copynames{F.S. de Boer, H.A. Hiep, S. de Gouw}    %% Fill in the first initial and last name of the authors
\begin{document}
%%%Note the beginning and end of the frontmatter section that starts here%%%%%
\begin{frontmatter}
  \title{Dynamic Separation Logic} 	   %%Title here and the
 %%%%%%%%%%%%%%%%%%%%%%%%%%%%			This Thanks is optional.
  %%%%Now the author(s) names(s)%%%%%
   \author{Frank S. de Boer\thanksref{a}\thanksref{b}\thanksref{frb}}	     %%Note NO SPACE between 
   \author{Hans-Dieter A. Hiep\thanksref{a}\thanksref{b}\thanksref{hdh}}		%last name and \thanksref{...} 
   \author{Stijn de Gouw\thanksref{c}\thanksref{sdg}}
    %%%Next come the addresses%%%%
   \address[a]{Leiden Institute for Advanced Computer Science (LIACS)\\Leiden University\\Leiden, the Netherlands}
   \address[b]{Computer Security group\\Centrum Wiskunde \& Informatica (CWI)\\Amsterdam, the Netherlands}
   \address[c]{Department of Computer Science\\Open University (OU)\\Heerlen, the Netherlands}
   \thanks[frb]{Email: \href{mailto:frb@cwi.nl}{\texttt{\normalshape
        frb@cwi.nl}}}
   \thanks[hdh]{Email:  \href{mailto:hdh@cwi.nl}{\texttt{\normalshape
        hdh@cwi.nl}}}
    \thanks[sdg]{Email:  \href{mailto:sdg@ou.nl}{\texttt{\normalshape
        sdg@ou.nl}}}
\begin{abstract} 
This paper introduces a dynamic logic extension of separation logic.
The assertion language of separation logic is extended with modalities for the
five  types of the basic instructions of separation logic: simple assignment,
look-up, mutation, allocation, and de-allocation.
The main novelty of the resulting dynamic logic is that it
allows to  combine different approaches to resolving these modalities.
One such approach is based on the standard weakest precondition calculus of separation logic.
The other approach introduced in this paper provides a novel alternative 
formalization   in the proposed dynamic logic extension of separation logic.
The soundness and completeness of this axiomatization has been formalized in the Coq theorem prover.

\end{abstract}
\end{frontmatter}

\section{Introduction}\label{intro}
This paper describes a study into the expressive power
of separation logic (SL, for short) with regard to the formalization of 
\emph{weakest preconditions} \cite{dijkstra1976}.
To this end, we introduce a novel dynamic logic extension of
SL, which we abbreviate by DSL (for Dynamic Separation Logic).

SL \cite{Reynolds02} extends Hoare logic
for the specification and verification of heap manipulating programs 
in terms of pre- and postconditions.
The assertion language of SL features
the basic heap assertion $(x\mapsto e)$,  `$x$ points to $e$', 
which expresses that the variable
$x$ denotes the single  allocated memory location which stores the value of the expression $e$.
The so-called separating conjunction  ($p\sep q$)
allows to split the heap, that is,
the set of allocated memory locations and their contents,
into two disjoint parts one of which satisfies the conjunct $p$ and the other
satisfies $q$.
The
separating implication $(p\sepimp{} q)$, roughly, holds
if every extension of the heap satisfies $q$,
whenever  $p$ holds for the extension itself (separately).
For an introduction to SL and an extensive survey of the literature, intended for a broad audience, see the paper by A.~Chargu\'eraud \cite{10.1145/3408998}.

Dynamic logic \cite{Harel79} generalizes Hoare logics by
introducing for each statement of the underlying programming language a
corresponding modality, so that the formula $[S]p$ expresses
the weakest precondition of the statement $S$ with respect to the postcondition $p$.
Informally, $[S]p$ is valid if every terminating computation establishes $p$.
%For example, the KeY tool \cite{DBLP:series/lncs/10001} uses
%first-order dynamic logic for  reasoning about Java programs.
In this paper we extend
the assertion language of SL  with \emph{modalities} for the
five  types of the basic instructions of SL: simple assignment,
look-up, mutation, allocation, and de-allocation.
For any such basic instruction $S$, we then can introduce in the Hoare logic
the axiom
$$
\HT{[S]p}{S}{p}
$$
which is trivially sound and complete by definition
of $[S]p$.
In case $S$ is a simple assignment $x:=e$
and $p$ is an assertion in standard SL, 
we can resolve 
the weakest precondition $[S]p$, as in first-order dynamic logic, simply by \emph{substituting} every free occurrence of $x$ in $p$ by the expression $e$.\footnote{After suitable renaming of the bound variables in $p$ such that no variable of $e$ gets bound.}
In SL we can resolve $[S]p$, for any other basic instruction $S$,
by a formula with a hole $C_S(\cdot)$ in SL itself, such that
$C_S(p)$ is equivalent to $[S]p$.
For example, the assertion
$$
(\exists y(x\mapsto y))* ((x\mapsto e)\sepimp p)
$$
states that the heap can be split in a sub-heap which consists of a single memory cell
denoted by $x$ such
that $p$ holds for every extension of the other part with a single memory
cell denoted by $x$ and which contains the value of $e$.
It follows that this assertion
is equivalent to $[[x]:=e]p$, where the \emph{mutation} instruction
$[x]:=e$ assigns the value of the expression $e$ to the heap location denoted by the variable $x$.

The main  contribution of this paper is a complementary approach to
resolving $[S]p$, for any  basic instruction.
In this approach we obtain an alternative characterization of the weakest precondition
$[S]p$ by a novel axiomatization  of the modalities
in DSL.
This axiomatization allows for
a characterization of $[S]p$ \emph{compositionally} in terms of the syntactical structure
of $p$.

O'Hearn, Reynolds, and Yang introduced  local axioms \cite{DBLP:conf/csl/OHearnRY01}
and show how to derive from these local axioms a weakest precondition
axiomatization of the basic instructions in SL,
using the frame rule and the separating implication for expressing the weakest
precondition. 
%From the completeness of the resulting  weakest precondition
%axiomatization one thus derives completeness of these local axioms.
%Yang \cite{yang2002semantic} shows that
However, the separating implication is actually not needed to prove completeness of the local axioms for simple assignments, look-up, allocation, and de-allocation.
We illustrate the expressiveness of DSL by
extending this result to the local mutation axiom.
We further illustrate the expressiveness of DSL by
a novel \emph{strongest postcondition} axiomatization.

Using the proof assistant
Coq, we have formally verified the  soundness and completeness proofs of the  axiomatization of the DSL modalities.
All our results can be readily extended to a programming language
involving (sequential) control structures such as loops.

\paragraph{Acknowledgments} The authors are grateful for the constructive feedback provided by the anonymous referees.

\section{Syntax and semantics}\label{sec:sl}
We follow the presentation of SL in
\cite{Reynolds02}.
A heap\footnote{All italicized variables are typical meta-variables, and we use primes and subscripts for other meta-variables of the same type, e.g.~$h$, $h'$, $h''$, $h_1$, $h_2$ are all heaps.} $h$ is represented by a (finitely-based) \emph{partial} function ${\mathbb{Z}\rightharpoonup \mathbb{Z}}$ and the domain of $h$ is denoted by $\mbox{\it dom}(h)$.
We write $h(n)=\bot$ if $n\not\in\mbox{\it dom}(h)$.
The heaps $h,h'$ are disjoint iff $\mbox{\it dom}(h)\cap \mbox{\it dom}(h')=\emptyset$.
A heap $h$ is partitioned in $h_1$ and $h_2$, denoted by $h=h_1\uplus h_2$, iff
$h_1$ and $h_2$ are disjoint,
$\mbox{\it dom}(h)= \mbox{\it dom}(h_1)\cup  \mbox{\it dom}(h_2)$,
and $h(n)=h_i(n)$ if $n\in\mbox{\it dom}(h_i)$ for $i\in\{1,2\}$.
%The latter definition implies $\mbox{\it dom}(h)\supseteq\mbox{\it dom}(h_1)\cup  \mbox{\it dom}(h_2)$.
%The following two properties follow:
%for any two disjoint heaps $h_1,h_2$ there always exists a unique heap $h$ such that $h=h_1\uplus h_2$,
%and for any heaps $h,h_1$ such that $h(n)=h_1(n)$ for all $n\in \mathit{dom}(h_1)$ there always exists a heap $h_2$ such that $h=h_1\uplus h_2$.

$V$ denotes a countably infinite set of integer variables, with typical element~$x$.
A store $s$ is a total function $V\rightarrow \mathbb{Z}$.
We abstract from the syntax of arithmetic expressions $e$,
and Boolean expressions $b$.
By $\mathit{var}(e)$ (resp. $\mathit{var}(b)$) we denote the finite set of variables that occur in $e$ (resp. $b$).
We have the Boolean constants $\T$ and $\F$, and $(e_1 = e_2)$ is a Boolean expression given arithmetic expressions $e_1$ and $e_2$.
By $s(e)$ we denote the integer value of $e$ in~$s$,
and by $s(b)$ we denote the Boolean value of $b$ in~$s$.
%Arithmetic Boolean expressions are denoted by $b$\textcolor{red}{,
%and are arithmetic expressions that express the truth value of $s(b) \neq 0$.}
%(it semantics $s(b)$ is a special case of $s(e)$).
Following \cite{Reynolds02} expressions thus do not refer
to the heap.
By $s[x:=v]$ and $h[n:=v]$ we denote the result of updating
the value of the variable $x$ and the location
$n$, respectively.
The definition of $h[n:=v]$ does not require
that $n\in \mbox{\it dom}(h)$.
More specifically, we have
$$
h[n:=v](m)=
\left\{
\begin{array}{ll}
v    & \mbox{if $n=m$} \\
h(m) &\mbox{otherwise}
%h(m)     & \mbox{if $m\in\mbox{\it dom}(h)\setminus\{n\}$}\\
%\mbox{\textcolor{red}{undefined}} & \mbox{\textcolor{red}{if $m %\not\in\mbox{\it dom}(h)\cup\{n\}$}}
\end{array}
\right.
$$
Thus, $\mbox{\it dom}(h[n:=v])=\mbox{\it dom}(h)\cup\{n\}$.
For heaps we also define the clearing of a location, denoted by $h[n:=\bot]$.
We have $h[n:=\bot](m)=\bot$ if $n = m$, and
$h[n:=\bot](m)=h(m)$ otherwise.
Similarly, we have $\mathit{dom}(h[n:=\bot])=\mathit{dom}(h)\setminus\{n\}$.

%We consider assignments of the forms $x:=e$, $x:=[e]$, $[x]:=e$, $x:=\new(e)$.
%and a program $S$ is either an assignment or the sequential composition $S_1;S_2$.
Following \cite{Reynolds02}, we 
have the following basic instructions:
%\footnote{Instructions like $[e']:=e$ can be  axiomatized in terms of the two assignments $x:=e'$ and $[x]:=e$ for some fresh $x$.}
$x:=e$ (simple assignment), $x:=[e]$ (look-up),
$[x]:=e$ (mutation), $x:=\new(e)$ (allocation), $\delete(x)$ (de-allocation).
Just like \cite{10.1145/373243.375719}, 
\begin{quote}
\emph{We will not give a full syntax of [statements], as the treatment of conditionals and looping
statements is standard. Instead, we will concentrate on assignment statements, which is where the main novelty of the approach lies.}
\end{quote}
The successful execution of any basic instruction $S$ is denoted by
$\langle S,h,s\rangle \Rightarrow (h',s')$,
whereas $\langle S,h,s\rangle \Rightarrow \mbox{\bf fail}$
denotes a failing execution (e.g. due to access of a `dangling pointer').
See Figure~\ref{fig:syntax} for their semantics (and see Appendix, Figure~\ref{fig:full-syntax}, for the full syntax and semantics).

\begin{figure}[t]
%\begin{framed}
$\langle x:=e,h,s\rangle \Rightarrow (h,s[x:=s(e)])$,\\
$\langle x:=[e],h,s\rangle \Rightarrow (h,s[x:=h(s(e))])$ if $s(e)\in\mbox{\it dom}(h)$,\\
$\langle x:=[e],h,s\rangle \Rightarrow \mbox{\bf fail}$ if $s(e)\not\in\mbox{\it dom}(h)$,\\
$\langle [x]:=e,h,s\rangle \Rightarrow (h[s(x):=s(e)],s)$ if $s(x)\in\mbox{\it dom}(h)$,\\
$\langle [x]:=e,h,s\rangle \Rightarrow \mbox{\bf fail}$ if $s(x)\not\in\mbox{\it dom}(h)$,\\
$\langle x:=\new(e),h,s\rangle \Rightarrow (h[n:=s(e)],s[x:=n])$ where $n\not\in\mbox{\it dom}(h)$.\\
$\langle \delete(x),h,s\rangle \Rightarrow (h[s(x):=\bot],s)$ if $s(x)\in\mbox{\it dom}(h)$,\\
$\langle \delete(x),h,s\rangle \Rightarrow \mbox{\bf fail}$ if $s(x)\not\in\mbox{\it dom}(h)$.
%\end{framed}
\caption{Semantics of basic instructions of heap manipulating programs.}
\label{fig:syntax}
\end{figure}

We follow \cite{10.1145/373243.375719} in the definition of the syntax and semantics of the assertion language of SL but we use a different atomic `weak points to' formula (as in \cite{reynolds2000intuitionistic} and \cite{DemriLM21}). In DSL we have additionally a modality for each statement $S$, which has
highest binding priority.
\[
p,q \Coloneqq b \mid (e\hookrightarrow e')  \mid (p\to q) \mid (\forall x p) \mid (p \sep q) \mid (p \sepimp q) \mid [S]p
\]
By $h,s\models p$ we denote the truth relation of classical SL, see Figure~\ref{fig:ass}.
Validity of $p$ is denoted by $\models p$. 
Semantics of DSL extends the semantics of SL by giving semantics to the modality, expressing the weakest precondition.
We further have the usual abbreviations:
%e.g., $\F$ denotes inconsistency, e.g., a Boolean formula $x\not=x$,
$\lnot p$ denotes ${(p\to \F)}$,
$(p\lor q)$ denotes $(\lnot p\to  q)$ (negation has binding priority over implication), $p\equiv q$ denotes
$(p\to q)\wedge (q\to p)$,
$(\exists x p)$ denotes $\lnot(\forall x(\lnot p))$
and note that $x$ is bound in $p$.
By logical connective we mean the connectives $\lnot,\land,\lor,\to,\forall,\exists$, and by separating connective we mean $\sep$ and $\sepimp$.
Further,
${(e\hookrightarrow -)}$ denotes ${\exists x(e \hookrightarrow x)}$ for a fresh $x$,
$\mbox{\bf emp}$ denotes ${\forall x(x\nothookrightarrow -)}$,
and ${(e\mapsto e')}$ denotes ${(e\hookrightarrow e')}\land (\forall x({(x\hookrightarrow -)}\to x=e))$ for a fresh~$x$.
We use $\not\hookrightarrow$ and $\not=$ as negations of the predicate as usual, and in particular $(e\nothookrightarrow -)$ is $\lnot\exists x(e \hookrightarrow x)$. We may drop matching parentheses if doing so would not give rise to ambiguity. Note that
$h,s\models \mbox{\bf emp}$ iff $\mbox{\it dom}(h)=\emptyset$,
and $h,s\models (e\mapsto e')$ iff $\mbox{\it dom}(h)=\{s(e)\}$ and $h(s(e))=s(e')$.
An assertion is \emph{first-order} if its construction does not involve separating connectives or modalities.

\begin{figure}[t]
%\begin{framed}
$h,s\models b$ iff $s(b)=\T$,\\
$h,s\models (e\hookrightarrow e')$ iff $s(e)\in \mbox{\it dom}(h)$ and $h(s(e))=s(e')$,\\
$h,s\models (p\land q)$ iff $h,s\models p$ and $h,s\models q$,\\
%$h,s\models (p\lor q)$ iff $h,s\models p$ or $h,s\models q$,\\
$h,s\models (p\to q)$ iff $h,s\models p$ implies $h,s\models q$,\\
%$h,s\models \exists x p$ iff $h,s[x:=n]\models p$ for some $n$,\\
$h,s\models (\forall x p)$ iff $h,s[x:=n]\models p$ for all $n$,\\
$h,s\models (p\sep q)$ iff $h_1,s\models p$ and $h_2,s\models q$ for some $h_1, h_2$ such that $h = h_1\uplus h_2$,\\
$h,s\models (p\sepimp q)$ iff $h',s \models p$ implies $h'',s\models q$ for all $h',h''$ such that $h'' = h\uplus h'$,\\
$h,s\models [S]p$ iff $\langle S,h,s\rangle\not\Rightarrow\mbox{\bf fail}$ and $h',s'\models p$ for all $h',s'$ such that $\langle S,h,s\rangle\Rightarrow(h',s')$.
%\end{framed}
\caption{Semantics of Dynamic Separation Logic.}
\label{fig:ass}
\end{figure}

The assertion ${(e\hookrightarrow e')}$ is implied by ${(e\mapsto e')}$, and to express the latter using the former requires the use of separating connectives (i.e. ${(e\hookrightarrow e')}$ is equivalent to ${\T\sep (e\mapsto e')}$), whereas our definition of ${(e\mapsto e')}$ requires only logical connectives, and thus we use ${(e\hookrightarrow e')}$ as atomic formula.

A specification $\HT{p}{S}{q}$ is a triple that consists of a precondition $p$, a program $S$, and a postcondition $q$.
Specifications are interpreted in the sense of strong partial correctness, which ensures absence of explicit failure.
Formally, following \cite{Reynolds02}, the validity of a specification, denoted $\models \HT{p}{S}{q}$, is defined as:
if $h,s\models p$, then $\langle S,h,s\rangle\not\Rightarrow \mbox{\bf fail}$ and also $\langle S,h,s\rangle\Rightarrow (h',s')$ implies $h',s'\models q$ for all $h',s'$.
%With this definition we abstract from out of memory errors, e.g. the specification $\HT{(\forall x)(x\hookrightarrow -)}{x:=\new(e)}{\F}$ is valid.

%We say a program $S$ simulates another program $S'$ if $\models \HT{p}{S}{q}$ implies $\models \HT{p}{S'}{q}$ for every assertions $p$, $q$. 
%As such, one could introduce the instruction $[e']:=e$ which can simulated by
%$x:=e';[x]:=e$ where $x$ is a fresh variable, and similar instructions can be introduced for $[e'] := [e]$, $[e'] := \new(e)$, and $\dispose{e}$.
%The ${\bf cons}$ and ${\bf dispose}$ constructs of \cite{Reynolds02} cannot be simulated with the formalism above,
%but we prefer our simpler presentation for technical convenience (see Section \ref{sec:extensions} for further discussion of extensions).

Note that we have that
$\models \HT{[S]q}{S}{q}$ holds,
%since $h,s\models [S]p$ implies $h',s'\models p$ where $\langle S,h,s\rangle\Rightarrow(h',s')$ for all $h',s'$ by the semantics of the modality. In particular, if $h,s\models [S]p$ holds then $\langle S,h,s\rangle$ does not lead to failure. Moreover, observe 
and $\models \HT{p}S{q}$ implies $\models p\to [S]q$, that is, $[S]q$ is the weakest precondition of statement $S$ and postcondition $q$.

\section{A sound and complete axiomatization of DSL}
In dynamic logic axioms are introduced to simplify formulas in which modalities occur. 
%In the first place, modalities with complex statements can be simplified in the standard manner:
%\begin{align*}
%[\mathbf{diverge}]p &\equiv \mbox{\bf true}\\
%[\mathbf{skip}]p &\equiv p\\
%[S_1;S_2]p &\equiv [S_1][S_2]p\\
%[\mathbf{if}\ b\ \mathbf{then}\ S_1\ \mathbf{else}\ S_2\ \mathbf{fi}]p &\equiv (b\to [S_1]p)\land(\lnot b\to [S_2]p)\\
%(q\land b\to [S]q)\text{ implies }[\mathbf{while}\ b\ \mathbf{do}\ %S\ \mathbf{od}]p &\equiv (q\land \lnot b\to p)
%\end{align*}
For example, we have the following basic equivalences
\mbox{\bf E1-3}
for simple assignments.

\begin{lemma}[Basic equivalences]\label{lem:BE}
Let $S$ denote a simple assignment $x:=e$ and $\circ$ denote
a (binary) logical or separating connective.
\begin{align*}
[S]\F &\equiv \F \tag{\mbox{\bf E1}}\\
[S](p\circ  q) &\equiv [S] p\circ  [S] q \tag{\mbox{\bf E2}}\\
[S](\forall y p) & \equiv \forall y ([S] p) \tag{\mbox{\bf E3}}\\
%[S](p\sep q) &\equiv [S]p \sep [S]q\\
%[S](p\sepimp q) &\equiv [S]p \sepimp [S]q
\end{align*}
In {\bf E3} we assume that $y$ does not
appear in $S$, neither in the left-hand-side of the assignment
$S$ nor in its right-hand-side.
\end{lemma}
The proofs of these equivalences proceed by a straightforward induction on
the structure of $p$, where the base cases of Boolean expressions and the weak points to predicate
are handled by a straightforward extension of the 
\emph{substitution lemma} for standard first-order logic.
By $b[e/x]$ we denote the result of replacing every occurrence of $x$ in 
the Boolean expression $b$ by the expression $e$ (and similar for arithmetic expressions).

\begin{lemma}[Substitution lemma]\label{lem:store-subst}
\begin{equation}
[x:=e]b \equiv b[e/x]\qquad [x := e](e'\hookrightarrow e'') \equiv (e'[e/x]\hookrightarrow e''[e/x])\tag{\mbox{\bf E4}}
\end{equation}
%for any Boolean expression $b$.
\end{lemma}

\begin{proof}
This lemma follows from the semantics of 
simple assignment modality and the substitution lemma
of first-order expressions: $s(e'[e/x])=s[x:=s(e)](e')$. Note that expressions do not refer to the heap.
%and the coincidence condition for expressions,
%$$\text{if }s(x)=s'(x) \text{ for all } x\in \mathit{var}(e) \text{ then } s(e)=s'(e).$$
%For example, we have\\
%$h,s\models (p*q)[x:=e]$
%\\iff (definition substitution)\\
%$h,s\models (p[x:=e])*(q[x:=e])$
%\\iff (semantics for separating conjunction)\\
%$h_1,s\models p[x:=e]\ \mbox{and}\ h_2,s\models q[x:=e]$, for some $h_1, h_2$ with %$h_1\uplus h_2=h$
%\\iff (induction hypothesis)\\
%$h_1,s[x:=s(e)]\models p\ \mbox{and}\ h_2,s[x:=s(e)]\models q$  for some $h_1, h_2$ %with $h_1\uplus h_2=h$
%\\iff (semantics for separating conjunction)\\
%$h, s[x:=s(e)]\models p*q$.
\end{proof}
%Note that by {\bf E1} and {\bf E2} we have
%$[x:=e](p\circ q)\equiv [x:=e]p\wedge [x:=e]q$, for any logical %connective
%$\circ$.

The above equivalences \mbox{\bf E1-3} do not hold in general
for the other basic instructions.
For example, we have
${[x:=[e]]\F\equiv \neg (e \hookrightarrow -)}$.
On the other hand, ${[x:=\new(0)]\F\equiv \F}$, but
${[x:=\new(0)](x\not=0)}$ is not equivalent to
$\neg ([x:=\new(0)](x=0))$, because $[x:=\new(0)](x\not=0)$
is equivalent to $(0\hookrightarrow -)$ (`zero is allocated'),
whereas $\neg ([x:=\new(0)](x=0)$) expresses that
$(n\not\hookrightarrow -)$, for some $n\not=0$
(which holds for any finite heap).
% TODO: Add example why E3 fails for other basic instructions.

The above equivalences \mbox{\bf E1-3}, with \mbox{\bf E2}
restricted to the (standard) logical connectives,
\emph{do} hold for the \emph{pseudo} instructions $\update{x} := e$,
a so-called \emph{heap update}, and
$\dispose{x}$, a so-called \emph{heap clear}.
These pseudo instructions are defined by the transitions
$$\langle \update{x} :=e,h,s\rangle \Rightarrow (h[s(x):=s(e)],s)
\mbox{ and }
\langle \dispose{x},h,s\rangle \Rightarrow (h[s(x):=\perp],s)$$
In contrast to  the mutation and de-allocation instructions, these pseudo-instructions do not require that $s(x)\in{\it dom}(h)$, e.g.,
if $s(x)\not\in{\it dom}(h)$ then the heap update $\update{x} :=e$ extends the domain of the heap,
whereas $[x] := e$ leads to failure in that case.
From a practical viewpoint, the heap update and heap clear pseudo-instructions are `lower level' instructions, e.g. in processors that implement virtual memory (where an operating system allocates memory on the fly whenever a program performs a write to a virtual address that is not allocated),
%and causes an interrupt), 
and on top of these instructions efficient memory allocation algorithms are implemented, e.g. malloc and free in C. %Separation logic was designed for reasoning at this higher level of abstraction, not at the lower level.
In the following lemma we give an axiomatization in DSL of the basic SL instructions
in terms of simple assignments and these two pseudo-instructions.
For comparison we also give the standard SL axiomatization
\cite{Reynolds02,FAISALALAMEEN201673,BannisterHK18}.

\begin{lemma}[Axioms basic instructions]\label{lem:ABI}
\begin{align*}
%\mbox{\bf E5} &[x:=[e]]p &\equiv & (e\hookrightarrow -)\wedge
%\exists y(y=e\wedge [x:=y]p)\\
[x:=[e]]p &\equiv \exists y((e\hookrightarrow y)\wedge [x:=y]p), \tag{\mbox{\bf E5}}\\
%[[x]:=e]p &\equiv (x\hookrightarrow -)\wedge [\update{x}:=e]p \tag{\mbox{\bf E6}}\\
%          &\equiv (x\mapsto -)\sep ((x\mapsto e) \sepimp p),\\
[[x]:=e]p & \equiv 
\left\{
\begin{array}{l}
(x\hookrightarrow -)\wedge [\update{x}:=e]p \tag{\mbox{\bf E6}}\\
(x\mapsto -)\sep ((x\mapsto e) \sepimp p)\\
\end{array}
\right.\\
%[x:=\new(e)]p &\equiv \forall x( (x\not\hookrightarrow-) \rightarrow [\update{x}:=e]p) \tag{\mbox{\bf E7}}\\
[x:=\new(e)]p &
\equiv 
\left\{
\begin{array}{l}
\forall x( (x\not\hookrightarrow-) \rightarrow [\update{x}:=e]p) \tag{\mbox{\bf E7}}\\
\forall x( (x\mapsto e) \sepimp p)\\
\end{array}
\right.\\
%[\delete(x)]p &\equiv (x\hookrightarrow -)\wedge [\dispose{x}]p \tag{\mbox{\bf E8}}\\
%&&\equiv (x\mapsto -) \sep p.
[\delete(x)]p &\equiv 
\left\{
\begin{array}{l}
(x\hookrightarrow -)\wedge [\dispose{x}]p \tag{\mbox{\bf E8}}\\
(x\mapsto -) \sep p
\end{array}
\right.
\end{align*}

Note that $[x:=y]p$ in \mbox{\bf E5} reduces
to $p[y/x]$ by \mbox{\bf E1-4}.
For technical convenience only, we require
in the axioms for ${x:=\new(e)}$ that
$x$ does not appear in $e$ (see Section~\ref{sec:extensions} to lift this restriction).
%Unlabeled equivalences are the standard weakest preconditions \cite{Reynolds02,FAISALALAMEEN201673,BannisterHK18}.
\end{lemma}

In the sequel  \mbox{\bf E5-8} refer to the corresponding
DSL equivalences.
The proofs of these equivalences are straightforward 
(consist simply of expanding the semantics of the involved modalities) and
therefore omitted.

We have the following  SL axiomatization of the 
heap update and heap clear pseudo-instructions.
\begin{align*}
[\update{x} := e]p &\equiv ((x\mapsto -)\sep ((x\mapsto e)\sepimp p))\lor ((x\nothookrightarrow -)\land ((x\mapsto e)\sepimp p))\\
[\dispose{x}]p &\equiv ((x\mapsto -)\sep p)\lor ((x\nothookrightarrow -)\land p)
\end{align*}
%Interestingly, these characterizations are complex: the size of the precondition grows considerably.
This axiomatization thus requires a case distinction between
whether or not $x$ is allocated.

For the complementary approach, we want to resolve the modalities
for the heap update and heap clear instructions compositionally in terms of $p$.
What thus remains for a complete axiomatization is a characterization of 
$[S]b$, $[S] (e\hookrightarrow e')$,
$[S](p\sep q)$, and $[S](p\sepimp q)$, where $S$ denotes 
one of the two  pseudo-instructions.
Lemma \ref{def:update} provides
an axiomatization in DSL of  a heap update.
%{\color{red} However, this instruction alone is not sufficient to simulate allocation.}

\begin{lemma}[Heap update]\label{def:update}
%Assuming the variables of $e$ and $x$ do not occur bound in $p$,
We have the following equivalences for the heap update modality.
\begin{align*}
[\update{x}:=e]b &\equiv  b, \tag{\mbox{\bf E9}}\\
%\item ${\bf emp}[\update{x}:=e]=\F$,
%\item $(y\mapsto e')[\update{x}:=e]= ({\bf emp}\vee x\mapsto -) \wedge y=x\wedge e'=e$
%\item $(e'\mapsto e'')[\update{x}:=e]= ({\bf emp} \vee x\mapsto -) \wedge x=e' \wedge e''=e$,
[\update{x}:=e](e'\hookrightarrow e'') &\equiv  (x=e' \wedge e''=e)\vee (x\not=e'\wedge e'\hookrightarrow e''), \tag{\mbox{\bf E10}}\\
[\update{x}:=e](p\sep q) &\equiv ([\update{x}:=e]p\sep q')\vee (p' \sep [\update{x}:=e]q), \tag{\mbox{\bf E11}}\\
%\item $(p*q)[\update{x}:=e]=  \overline{p}*\overline{q}$,\\
%where $\overline{p}$ (and similarly $\overline{q}$) is defined by\\ 
%$\ITE{x\hookrightarrow}{p[\update{x}:=e]}{p}$
[\update{x}:=e](p\sepimp q) &\equiv  p'\sepimp [\update{x}:=e]q, \tag{\mbox{\bf E12}}
\end{align*}
where $p'$ abbreviates $p\wedge (x\not\hookrightarrow-)$ and, similarly, $q'$ abbreviates $q\wedge (x\not\hookrightarrow-)$.
%\item $(p \sepimp{} q)[\update{x}:=e]=   (p\wedge \neg (x\hookrightarrow %-)) \sepimp{} q[\update{x}:=e]$
\end{lemma}

These equivalences we can informally explain as follows.
Since the heap update $\update{x}:=e$ does not affect the store, and the evaluation of a Boolean condition
$b$ only depends on the store, we have that $([\update{x}:=e]b)\equiv  b$.

Predicting whether ${(e'\hookrightarrow e'')}$ holds after $\update{x}:=e$, we only
need to make a distinction between whether $x$ and $e'$ are aliases, that is, whether they denote
the same location, which is simply expressed by $x=e'$.
If $x=e'$ then $e''=e$ should hold, otherwise
$(e'\hookrightarrow e'')$ (note again, that $\update{x}:=e$ does not
affect the values of the expressions $e,e'$ and $e''$).
As a basic example, we compute
$$
\begin{array}{lll}
[\update{x}:=e](y\hookrightarrow -) &\equiv & (\mbox{definition }
y\hookrightarrow -)\\

[\update{x}:=e]\exists z (y\hookrightarrow z) &\equiv &(\mbox{\bf E3})\\
\exists z [\update{x}:=e](y\hookrightarrow z) &\equiv &(\mbox{\bf E10})\\
\exists z((y=x\wedge e=z)\vee (y\not=x\wedge (y\hookrightarrow z)))
&\equiv & (\mbox{semantics SL})\\
y\not=x\to (y\hookrightarrow -)
\end{array}
$$
We use this derived equivalence in the following example:
$$
\begin{array}{lll}
[\update{x}:=e](y\mapsto-) &\equiv & (\mbox{definition $y\mapsto-$})\\

[\update{x}:=e]((y\hookrightarrow -) \wedge 
\forall z ( (z\hookrightarrow -)\to z=y))&\equiv &
(\mbox{\bf E2, E3, E9})\\

[\update{x}:=e](y\hookrightarrow -) \wedge 
\forall z  ([\update{x}:=e](z\hookrightarrow -)\to z=y)
&\equiv & (\mbox{see above})\\
(y\not=x\to (y\hookrightarrow -))\wedge
\forall z ((z\not=x\to (z\hookrightarrow -))\to z=y)  &\equiv &
(\mbox{semantics SL})\\
%y=x\wedge (\forall z((z\hookrightarrow -)\to z=x)) &\equiv &(\mbox{semantics SL})\\
y=x\wedge ( {\bf emp}\vee (x\mapsto -))
\end{array}
$$

Predicting whether $(p\sep q)$ holds after the heap update $\update{x}:=e$, we need to distinguish
between whether $p$ or $q$ holds for the sub-heap that contains the (updated) location $x$.
Since we do not assume that $x$ is already allocated, we instead distinguish between  
whether $p$ or $q$ holds initially for the sub-heap that  does \emph{not} contain the updated location $x$.
As a simple example, we compute
%$$
%\begin{array}{lll}
%[\update{x}:=e]((y\mapsto-)\sep (z\mapsto -)) &\equiv & 
%(\mbox{\bf E11})\\
%([\update{x}:=e](y\mapsto-)\sep ((z\mapsto -)\wedge %(x\nothookrightarrow -))) \vee 
%(((y\mapsto -)\wedge (x\nothookrightarrow -))
%\sep [\update{x}:=e](z\mapsto-) &\equiv & \mbox{see above}\\

%\end{array}
%$$
$$
\begin{array}{lll}
[\update{x}:=e](\T\sep (x\mapsto-)) &\equiv & (\mbox{\bf E9,E11})\\
(\T \sep ((x\mapsto-)\wedge (x\not\hookrightarrow-)))
\vee ((x\not\hookrightarrow-)\sep [\update{x}:=e](x\mapsto-) &\equiv & (\mbox{see above})\\
(\T \sep ((x\mapsto-)\wedge (x\not\hookrightarrow-)))
\vee ((x\not\hookrightarrow-)\sep ( {\bf emp}\vee (x\mapsto -))) &\equiv & (\mbox{semantics SL}) \\
(\T\sep\F) \vee  ((x\not\hookrightarrow-)\sep ( {\bf emp}\vee (x\mapsto -))) &\equiv & (\mbox{semantics SL}) \\
\T
\end{array}
$$
Note that this coincides with the above
calculation of $[\update{x}:=e](y\hookrightarrow -)$,
which also reduces to $\T$, instantiating $y$ by $x$.

The semantics of $(p\sepimp q)$ after the heap update $\update{x}:=e$ involves universal quantification
over all disjoint heaps that  do not contain $x$ (because after the heap update $x$ is allocated).
Therefore we simply add the condition that $x$ is not allocated to $p$, and apply the heap update
to $q$.
As a very basic example, we compute
$$
\begin{array}{lll}
[\update{x}:=0]((y\hookrightarrow 1)\sepimp (y\hookrightarrow 1)) &\equiv &
(\mbox{\bf E12})\\
((y\mapsto 1)\wedge (x\not\hookrightarrow -))\sepimp [\update{x}:=0](y\hookrightarrow 1)) &\equiv &
(\mbox{\bf E10})\\
((y\mapsto 1)\wedge (x\not\hookrightarrow -))\sepimp 
((y=x\wedge 0=1)\vee (y\not=x\wedge y\hookrightarrow 1))&\equiv &
(\mbox{semantics SL})\\
\T
\end{array}
$$
Note that $(y\hookrightarrow 1)\sepimp (y\hookrightarrow 1)\equiv \T$ and $[\update{x}:=0]\T\equiv \T$.
%By renaming bound variables of $p$, it is always possible to find an alphabetic variant for which this substitution is defined, since $\mathit{var}(e)\cup\{x\}$ is finite.

%\begin{proof}[Proof of Lemma~\ref{def:update}.]
\vspace{4pt}\noindent {\bf Proof of Lemma~\ref{def:update}.}
%${\bf emp}$: It suffices to observe 
%that $h[s(x):=v],s  \not\models {\bf emp}$.
%The proof proceeds by induction on the structure of $p$.
%Here we go.
\begin{description}
%\item $h,s  \models b[\update{x}:=e]$
%\\iff (semantics boolean expressions: they do not depend on the heap)\\
%$s(b[\update{x}:=e])=\T$
%\\iff (Definition~\ref{def:update}, substitution for boolean expressions)\\
%$s(b)=\T$
%\\iff (semantics boolean expressions)\\
%$h[s(x):=s(e)],s  \models b$.
\item[\mbox{\bf E9}] $h,s\models [\update{x}:=e]b$ \\
iff (semantics  heap update modality)\\
$h[s(x):=s(e)],s\models b$\\
iff ($b$ does not depend on the heap)\\
$h,s\models b$
%\item $h, s \models {\bf emp}[\update{x}:=e]$
%iff (definition  substitution)\\
%$h, s \models \F$
%iff (semantics assertions: no heap/state pair satisfies $\F$)\\
%$h[s(x):=s(e)], s \models \F$
%iff (${\it dom}(h[s(x):=s(e)])\not=\emptyset$) \\
%$h[s(x):=s(e)], s \models {\bf emp}$

%For the last step, observe that in SL there are no undefined arithmetic expressions $e$ such that $s(e)=\bot$. Hence $s(x)\in \mbox{\it dom}(h[s(x):=s(e)])$, thus
%$h[s(x):=s(e)],s  \not \models {\bf emp}$.

%\item $h,s \models (e'\mapsto e'')[\update{x}:=e]$
%\\iff (definition substitution for points-to)\\
%$h,s  \models ({\bf emp} \vee x\mapsto -) \wedge x=e' \wedge e''=e$
%\\iff (semantics assertions)\\
%$\mbox{\it dom}(h)\subseteq \{s(x)\}$,  $s(x)=s(e')$
%and  $s(e)=s(e'')$
%\\iff (definition $h[s(x):=s(e)]$)\\
%$\mbox{\it dom}(h[s(x):=s(e)])=\{s(x)\}$ and $s(x)=s(e')$, and $s(e)=s(e'')$
%\\iff (definition $h[s(x):=s(e)]$)\\
%$\mbox{\it dom}(h[s(x):=s(e)])=\{s(x)\}$, $h[s(x):=s(e)](s(e')) = s(e'')$
%\\iff (semantics points-to)\\
%$h[s(x):=s(e)],s  \models e'\mapsto e''$.

\item[\mbox{\bf E10}] $h,s \models [\update{x}:=e](e'\hookrightarrow e'')$\\
iff (semantics  heap update modality)\\
$h[s(x):=s(e)],s\models e'\hookrightarrow e''$\\
iff (semantics points-to)\\
$h[s(x):=s(e)](s(e')) = s(e'')$\\
iff (definition $h[s(x):=s(e)]$)\\
if $s(x)=s(e')$ then  $s(e)=s(e'')$ else $h(s(e'))=s(e'')$\\
iff (semantics assertions)\\
$h,s  \models  (x=e' \wedge e''=e) \vee (x\not=e'\wedge e'\hookrightarrow e'')$
\item[\mbox{\bf E11}] 
$h,s  \models [\update{x}:=e](p*q)$\\
iff (semantics  heap update modality)\\
$h[s(x):=s(e)],s\models p\sep q$.\\
From here we proceed as follows.
By the semantics of separating conjunction, there exist
$h_1$ and $h_2$  such that $h[s(x):=s(e)]=h_1\uplus h_2$,
$h_1,s \models p$, and 
$h_2,s  \models q$. Let $s(x)\in {\it dom}(h_1)$
(the other case runs similarly).
So $h[s(x):=s(e)]=h_1\uplus h_2$ implies $h_1(s(x))=s(e)$
and $h=h_1[s(x):=h(x)]\uplus h_2$,
By the semantics of the heap update modality, $h_1(s(x))=s(e)$
and $h_1,s \models p$ implies
$h_1[s(x):=h(x)],s\models [\update{x}:=e]p$.
Since $s(x)\not\in {\it dom}(h_2)$, we have
$h_2,s  \models q \wedge x\not\hookrightarrow-$.
By the semantics of separation conjunction
we conclude that $h,s\models [\update{x}:=e]p \sep q'$
($q'$ denotes $q \wedge x\not\hookrightarrow-$).

In the other direction, from 
$h,s\models [\update{x}:=e]p \sep q'$ (the other case runs similarly) we derive
that there exist
$h_1$ and $h_2$  such that $h=h_1\uplus h_2$,
$h_1,s \models [\update{x}:=e] p$ and
$h_2,s  \models q'$.
By the semantics of the heap update modality it follows that
$h_1[s(x):=s(e)],s \models p$.
Since $s(x)\not\in{\it dom}(h_2)$, we have that
$h[s(x):=s(e)]=h_1[s(x):=s(e)]\uplus h_2$,
and so $h[s(x):=s(e)],s\models p\sep q$, that is,
$h,s\models [\update{x}:=e] (p\sep q)$.
\item[\mbox{\bf E12}] $h,s  \models [\update{x}:=e](p \sepimp q)$\\
iff (semantics of heap update modality)\\
$h[s(x):=s(e)],s  \models p \sepimp{} q$\\
iff (semantics separating implication)\\
for every $h'$ disjoint from $h[s(x):=s(e)]$:
if $h',s  \models p$
then $h[s(x):=s(e)]\uplus h',s  \models q$\\
iff (since $s(x)\not\in{\it dom}(h')$)\\
for every $h'$ disjoint from $h$:
if $h',s  \models p\wedge x\not\hookrightarrow-$
then $(h\uplus h')[s(x):=s(e)],s  \models q$\\
iff (semantics of heap update modality)\\
for every $h'$ disjoint from $h$:
if $h',s  \models p\wedge x\not\hookrightarrow-$
then $h\uplus h',s  \models [s(x):=s(e)]q$\\
iff (semantics separating implication)\\
$h,s \models (p \wedge x\not\hookrightarrow-) \sepimp [\update{x}:=e]q$.
\hfill$\square$
\end{description}
%\end{proof}

%TODO SDG commented the next paragraph out could be added back in, esp if Hans proves the given equivalence in Coq
%{\bf Note: here it would be interesting to provide a formal proof in SL that 
%$(x\hookrightarrow-)\wedge ((y=x\wedge z=0)\vee (y\not=x\wedge y\hookrightarrow z))$ and
%$x\hookrightarrow-\sep (x\mapsto 0\sepimp y\hookrightarrow z)$ are equivalent.
%Also of interest is how to derive our above instance of the backwards mutation axiom in SL
%using the other mutation axioms, plus possibly the frame rule (I don't see immediately how).}

%Let $s(x)=0$, $s(y)=1$, and $s(z)=1$.
%Further, let ${\it dom}(h)=\{1\}$ and $h(1)=1$.
%To decide 
%$h,s\models x\mapsto 0\sepimp y\hookrightarrow z$, we 
%instantiate proposition 1 (\cite{})):
%$X=\{x,y,z\}$, $|x\mapsto 0|=1$, $|y\hookrightarrow z|=
%|y\mapsto z *\T=|y\mapsto z|+|\T|=1+0=1$.
%So $B=\{2\}$. We take  $v=2$.
%Then we have to check $h\uplus h',s \models y\hookrightarrow z $, for all $h'$
%such that 
%\begin{itemize}
%\item $h$ and $h'$ are disjoint,
%\item $h',s\models x\mapsto 0$,
%\item ${\it dom}(h')\subseteq \{0,1,2\}$, 
%\item $h'(n)\in\{ 0,1,2,\perp\}$, for $n\in {\it %dom}(h')$.
%\end{itemize}

%In the axiomatization of the allocation instruction
%${x:=\new(e)}$ we assume the variables of $e$ and $x$ do not occur bound in $p$, but also that $x$ does not occur in $e$.
%The case in which $x$ does occur in $e$ can be handled by first storing the value $e$ in a temporary variable.

The equivalences for the heap clear modality 
in the following lemma can be informally explained as follows:
Since  $\dispose{x}$ does not affect the store, and the evaluation of a Boolean condition
$b$ only depends on the store, we have that $[\dispose{x}]b = b$.
%After executing $\dispose{x}$ the heap is empty
%if and only if the heap before at most  contains $x$,
%so $\mbox{\bf emp}[\dispose{x}] \equiv \mbox{\bf emp}\vee (x\mapsto -)$.
For $e\hookrightarrow  e'$ to hold after executing $\dispose{x}$, we must initially have that $x\not=e$
and $e\hookrightarrow  e'$.
As a simple example, we have that
$\forall y,z (y\not\hookrightarrow z)$ characterizes the empty heap.
It follows that
$[\dispose{x}](\forall y,z ( y\not\hookrightarrow z))$ is equivalent to
$\forall y,z (\neg (y\not= x \wedge y\hookrightarrow z))$.
The latter first-order formula
is equivalent to
$\forall y,z (y=x \vee  y\not\hookrightarrow z)$.
This assertion thus states
that the domain  consists at most of the location $x$,
which indeed ensures that after  $\dispose{x}$
the heap is empty.
%For $e\mapsto e'$ to hold after executing $\dispose{x}$, we must initially have that %$x\not=e$,
%$e\hookrightarrow e'$, and that initially at most $e$ and $x$ are allocated: thus
%$(e\mapsto e')[\dispose{x}] \equiv (e\hookrightarrow e') \wedge x\not= e \wedge
%\forall y((y\hookrightarrow -) \to (y=e\vee y=x))$.
To ensure that $p\sep q$ holds after clearing $x$ it suffices
to show that the initial heap can be split such that
both $p$ and $q$ hold in their respective sub-heaps
with $x$ cleared.
The semantics of $p\sepimp q$ after clearing $x$ involves universal quantification
over all disjoint heaps that  do may contain $x$, whereas before executing $\dispose{x}$ it
involves universal quantification
over all disjoint heaps that  do \emph{not} contain $x$, in case $x$ is  allocated initially.
To formalize  in the initial configuration   universal quantification
over all disjoint heaps we distinguish between all disjoint heaps that do not contain $x$
and \emph{simulate} all disjoint heaps
that   contain $x$ by interpreting  both $p$ and $q$ in $p\sepimp q$
in the context of heap updates $\update{x}:=y$
with \emph{arbitrary} values $y$ for the location $x$.
As a very basic example, consider
$[\dispose{x}]((x\hookrightarrow 0)\sepimp (x\hookrightarrow 0))$,
which should be equivalent to $\T$.
The left conjunct 
$((x\hookrightarrow 0)\wedge (x\not\hookrightarrow -))\sepimp [\dispose{x}](x\hookrightarrow 0))$
of the resulting formula after applying \mbox{\bf E16} 
is equivalent to $\T$ (because $(x\hookrightarrow 0)\wedge (x\not\hookrightarrow -)$
is equivalent to $\F$).
We compute the second conjunct (in the application of \mbox{\bf E10}
we omitted  some trivial reasoning steps):
$$
\begin{array}{lll}
\forall y ([\update{x}:=y](x\hookrightarrow 0)\sepimp [\update{x}:=y](x\hookrightarrow 0) &\equiv & (\mbox{\bf E10})\\
\forall y (y=0\sepimp y=0) &\equiv & (\mbox{semantics SL})\\
\T
\end{array}
$$

\begin{lemma}[Heap clear]\label{def:dispose-subst-CSL}
We have the following equivalences for the heap clear modality.
\begin{align*}
[\dispose{x}]b &\equiv  b, \tag{\mbox{\bf E13}}\\
%\item ${\bf emp}[\dispose{x}]= {\bf emp}\vee (x\mapsto -)$,
%\item $(e\mapsto e')[\dispose{x}]= 
%e\hookrightarrow e' \wedge x\not= e\wedge \forall y
%((y\hookrightarrow -) \to (y=e\vee y=x))$
%$(e'\mapsto e'')[\dispose{x}]= ((e'\mapsto e'' )* ({\bf emp}\vee x\hookrightarrow)) \wedge x\not=e'$
[\dispose{x}](e\hookrightarrow e') &\equiv  (x\not= e)\wedge (e\hookrightarrow e'), \tag{\mbox{\bf E14}}\\
[\dispose{x}](p*q) &\equiv [\dispose{x}]p \sep [\dispose{x}]q, \tag{\mbox{\bf E15}}\\
%\vee ((q[\dispose{x}]\wedge (x\hookrightarrow -))*p)$
%\item $(p\sep q)[[x]:=\perp]= p[[x]:=\perp]\sep q[[x]:=\perp]$
[\dispose{x}](p \sepimp q) &\equiv ((p\wedge x\not\hookrightarrow -) \sepimp [\dispose{x}]q)\wedge \forall y ([\update{x}:=y]p\sepimp{} [\update{x}:=y]q), \tag{\mbox{\bf E16}}
\end{align*}
where $y$ is fresh.
\end{lemma}

%For technical convenience and w.l.o.g.. we restrict in the above
%definition of $p[[x]:=\perp]$ to assertions $p$ which do not
%contain basic assertions ${\bf emp}$ and $e\mapsto e'$
%(which can be expressed by $e\hookrightarrow e$).

%$(\exists y\forall z (z\hookrightarrow - \to y=z))[[x]:=\perp]$,
%which states that at most one location is allocated,
%reduces to
%$\exists y\forall z ((z\not=x\wedge z\hookrightarrow -) \to y=z))$,
%which states that at most one location different from $x$  is allocated, that is, at most two locations are  allocated.

\begin{proof}
Here we go.
\begin{description}
\item[\mbox{\bf E13}]
$[\dispose{x}]b\equiv  b$.  As above, it suffices to observe that the evaluation of $b$ does not depend on the heap.
%\item
%$h,s\models  {\bf emp}[\dispose{x}]$ 
%\mbox{iff} (definition substitution)\\
%$h,s\models {\bf emp}\vee x\mapsto -$ 
%\mbox{iff} (semantics assertions)\\
%${\it dom}(h)\subseteq \{s(x)\}$
%\mbox{iff} (basic set-theory, definition $h[\dispose{s(x)}]$)\\
%${\it dom}(h[\dispose{s(x)}])=\emptyset$
%\mbox{iff} (semantics assertions)\\
%$h[\dispose{s(x)}],s\models {\bf emp}$.
%\item
%$h,s\models (e\mapsto e')[\dispose{x}]$
%\mbox{iff} (definition substitution)\\
%$h,s\models e\hookrightarrow e' \wedge x\not= e\wedge \forall y
%((y\hookrightarrow -) \to (y=e\vee y=x))$
%\mbox{iff} (semantics assertions)\\
%$h(s(e))=s(e')$, $s(x)\not=s(e)$, and $\{s(e)\}\subseteq  {\it dom}(h)\subseteq %\{s(x),s(e)\}$
%\mbox{iff} (basic set-theory, definition $h[\dispose{s(x)}]$) \\
%$h[\dispose{s(x)}](s(e))=s(e')$ and ${\it dom}(h[\dispose{s(x)}])=\{s(e)\}$
%\mbox{iff} (semantics assertions)\\
%$h[\dispose{s(x)}],s\models e\mapsto e$.
\item[\mbox{\bf E14}]
$h,s  \models [\dispose{x}](e\hookrightarrow e')$\\
iff (semantics heap clear modality)\\
$h[\dispose{s(x)}],s  \models e\hookrightarrow e'$\\
iff (semantics points-to)\\
$s(e)\in{\it dom}(h[\dispose{s(x)}])$ and $h[\dispose{s(x)}](s(e))=h(s(e))=s(e')$\\
iff (semantics assertions)\\
$h,s\models  x\not= e\wedge e\hookrightarrow e'$

\item[\mbox{\bf E15}]
$h,s  \models [\dispose{x}](p*q)$\\
iff (semantics heap clear modality)\\
$h[\dispose{s(x)}],s\models p\sep q$\\
iff (semantics separating conjunction)\\
$h_1,s\models p$ and $h_2,s\models q$,
for some $h_1,h_2$ such that $h[\dispose{s(x)}]=h_1\uplus h_2$\\
iff (semantics heap clear modality)\\
$h_1,s\models [\dispose{x}]p$ and $h_2,s\models [\dispose{x}]q$,
for some $h_1,h_2$ such that $h=h_1\uplus h_2$.\\
Note: $h=h_1\uplus h_2$ implies $h[\dispose{s(x)}]=h_1[\dispose{s(x)}]\uplus h_2[\dispose{s(x)}]$, and, conversely, $h[\dispose{s(x)}]=h_1 \uplus h_2$ implies there exists $h'_1,h'_2$ such that $h=h'_1\uplus h'_2$ and $h_1=h'_1[\dispose{s(x)}]$ and $h_2=h'_2[\dispose{s(x)}]$.

\item[\mbox{\bf E16}]
$h,s  \models [\dispose{x}](p \sepimp{} q)$\\
\mbox{iff} (semantics heap clear modality)\\
$h[s(x):=\perp],s  \models p \sepimp{} q$.\\
From here we proceed as follows.
First we show that 
$h,s \models  ((p \wedge x\not\hookrightarrow -)\sepimp{} [\dispose{x}]q)$
and
$h,s\models {\forall y ([\update{x}:=y]p\sepimp{} [\update{x}:=y]q)}$
implies $h[s(x):=\perp],s  \models p \sepimp{} q$.
Let $h'$ be disjoint from $h[s(x):=\perp]$
and $h',s\models p$. We have to show  that
$h[s(x):=\perp]\uplus h',s\models q$.
We distinguish the following two cases.

\begin{itemize}
    \item First, let  $s(x)\in {\it dom}(h')$.
We then introduce ${s'=s[y:=h'(s(x))]}$.
We have $h',s'\models p$ (since $y$ does not occur in $p$), so it follows by the semantics of the heap update modality
that $h'[s(x):=\bot],s'\models [\update{x}:=y]p$.
Since $h'[s(x):=\bot]$ and $h$ are disjoint
(which clearly follows from that $h'$ and  $h[s(x):=\bot]$
are disjoint), 
and since $h,s'\models [\update{x}:=y]p\sepimp{} [\update{x}:=y]q$,
we have that $h\uplus (h'[s(x):=\bot]), s'\models [\update{x}:=y]q$.
Applying again the semantics of the heap update modality,
we obtain 
$(h\uplus (h'[s(x):=\bot]))[s(x):=s'(y)], s'\models q$.
We then can conclude this case observing  that $y$ does not occur in $q$ and that
$h[s(x):=\bot]\uplus h'=(h\uplus (h'[s(x):=\bot]))[s(x):=s'(y)]$.

    \item Next, let $s(x)\not\in {\it dom}(h')$.
So $h'$ and $h$ are disjoint, and thus
(since $h,s \models  (p \wedge x\not\hookrightarrow -) \sepimp{} [\dispose{x}]q$) we have
$h\uplus h',s\models [\dispose{x}]q$.
From which we derive $(h\uplus h')[s(x):=\bot],s\models q$
by the induction hypothesis.
We then can conclude this case by the observation that
${h[s(x):=\bot]\uplus h'}={(h\uplus h')[s(x):=\bot]}$.
\end{itemize}

Conversely, assuming 
$h[s(x):=\bot],s  \models p \sepimp{} q,$
we first show that
$h,s \models  (p\wedge x\not\hookrightarrow -) \sepimp{} [\dispose{x}] q$
and then
$h,s\models \forall y
([\update{x}:=y]p\sepimp{} [\update{x}:=y]q).$
\begin{itemize}
    \item Let $h'$ be disjoint from $h$ and $h',s\models p \wedge x\not\hookrightarrow -$.
We have to show that $h\uplus h',s\models [\dispose{x}]q$, that is,
$(h\uplus h')[s(x):=\perp],s\models q$ (by the semantics of the heap clear update).
Clearly, $h[s(x):=\bot]$ and $h'$ are disjoint,
and so  $h[s(x):=\bot]\uplus h',s\models q$ follows from our assumption.
We then can conclude this case by the observation that
$(h\uplus h')[s(x):=\bot]=h[s(x):=\bot]\uplus h'$,
because $s(x)\not\in {\it dom}(h')$.

    \item Let $h'$ be disjoint
from $h$ and $s'=s[y:=n]$,
for some $n$ such that $h',s'\models [\update{x}:=y]p$.
We have to show that
$h\uplus h',s'\models [\update{x}:=y]q$.
By the semantics of the heap update modality it follows that
$h'[s(x):=n],s'\models p$, that is,
$h'[s(x):=n],s\models p$ (since $y$ does not occur in $p$).
Since $h'[s(x):=n]$ and  $h[s(x):=\bot]$ are disjoint,
we derive from the assumption $h[s(x):=\bot],s  \models p \sepimp{} q$
that $h[s(x):=\bot]\uplus h'[s(x):=n],s\models q$.
Again by the semantics of the heap update modality we have
that $h\uplus h',s'\models [\update{x}:=y]q$ iff
$(h\uplus h')[s(x):=n],s'\models q$ 
(that is, $(h\uplus h')[s(x):=n],s\models q$, because
$y$ does not occur in $q$).
We then can conclude this case by the observation that
${(h\uplus h')[s(x):=n]}={h[s(x):=\bot]\uplus h'[s(x):=n]}$.
\end{itemize}
\end{description}
\end{proof}

We denote by $\mbox{\bf E}$ the \emph{rewrite system}
obtained from the equivalences $\mbox{\bf E1-16}$
by orienting these equivalences from left to right, e.g.,
equivalence \mbox{\bf  E1} is turned into a rewrite
rule $[S]\F\Rightarrow \F$.
The following theorem states that the rewrite system
\mbox{\bf E} is complete, that is, confluent and strongly normalizing.
Its proof is straightforward (using standard techniques) and therefore omitted.

\begin{theorem}[Completeness of \mbox{\bf  E}]$\;\;\;$\label{thm:completeness}
\begin{itemize}
\item {\bf Normal form.}
Every standard formula $p$ of SL is in normal form
(which means that it cannot be reduced by the rewrite system
$\mbox{\bf E}$).
\item {\bf Local confluence.}
For any two reductions $p\Rightarrow q_1$ and $p\Rightarrow q_2$
($p$ a formula of DSL)
there exists a DSL formula $q$ such that
$q_1\Rightarrow q$ and $q_2\Rightarrow q$.
\item {\bf Termination.}
There does not exist an infinite chain of reductions
$p_1\Rightarrow p_2\Rightarrow p_3 \cdots$.
\end{itemize}
\end{theorem}

We now show an example of the interplay between the modalities for heap update and heap clear.
We want to derive
\[
\HT{\forall x((x\not\hookrightarrow-) \to p)}{x := \new(0); \delete(x)}{p}
\]
where statement $x := \new(0); \delete(x)$ simulates the so-called random assignment \cite{Harel79}: the program terminates with a value of $x$ that is chosen non-deterministically.
First we apply the axiom \mbox{\bf E8} for de-allocation to obtain
\[
\HT{(x\hookrightarrow -)\land [\dispose{x}]p}{\delete(x)}{p}.
\]
Next, we apply the axiom \mbox{\bf E8} for allocation to obtain 
\[
\begin{array}{c}
\{\forall x((x\nothookrightarrow -)\to [\update{x} := 0]((x\hookrightarrow -)\land [\dispose{x}]p))\}\\
{x := \new(0)}\\
\{(x\hookrightarrow -)\land p[\dispose{x}]\}.
\end{array}
\]
Applying  \mbox{\bf E10} (after pushing the heap update modality inside), 
followed by some basic first-order reasoning,
we can reduce $[\update{x} := 0](\exists y(x\hookrightarrow y))$ 
to true.
So we obtain
\[
\begin{array}{c}
\{\forall x((x\nothookrightarrow -)\to [\update{x} := 0][\dispose{x}]p)\}\\
{x := \new(0)}\\
\{(x\hookrightarrow -)\land p[\dispose{x}]\}.
\end{array}
\]
In order to proceed we formalize the interplay between the modalities for heap update and heap clear by the following general equivalence:
%\begin{align*}
%[\update{x} := e][\dispose{x}]p &\equiv [\dispose{x}]p %\tag{\mbox{\bf E17}}\\
%\end{align*}
$$
[\update{x} := e][\dispose{x}]p \equiv [\dispose{x}]p$$
We then complete the proof by
applying the sequential composition rule and consequence rule, using the above
equivalence  and the following axiomatization of the heap clear modality:
%\begin{align*}
%(x\not\hookrightarrow -)\wedge [\dispose{x}]p &\equiv 
%(x\not\hookrightarrow -)\wedge p \tag{\mbox{\bf E18}}\\
%\end{align*}
$$
(x\not\hookrightarrow -)\wedge [\dispose{x}]p \equiv 
(x\not\hookrightarrow -)\wedge p$$

The above axiomatization can be extended in the standard manner
to a program logic for sequential while programs, see \cite{Harel79}, which does not require the frame rule,
nor any other adaptation rule besides the consequence rule.
For recursive programs however one does need more adaptation rules:
a further discussion about the use of the frame rule
in a completeness proof for recursive programs is outside 
the scope of this paper.

\section{Expressiveness DSL}
%From the completeness of the resulting  weakest precondition
%axiomatization one thus derives completeness of these local axioms.
In this section, we illustrate the expressiveness of 
DSL
in a  completeness proof of the local mutation axiom
and a novel strongest postcondition axiomatization.

\subsection{Completeness local axioms}
We consider the completeness of the following local mutation axiom
(completeness of the local axioms for the other standard basic instructions have already been established, as observed in the Introduction)
$$
    \HT{x\mapsto -}{[x]:=e}{x\mapsto e}
$$
The proof itself does not make use of the separating implication.

\begin{theorem}[Completeness local mutation axiom]
   If $\models \HT{p}{[x]:=e}{q}$ then $\HT{p}{[x]:=e}{q}$ is derivable using the local mutation axiom, frame rule, and consequence rule.
\end{theorem}
\begin{proof}
The problem here is how to compute a `frame' $r$ for a given valid
specification $\HT{p}{[x]:=e}{q}$
so that $p$ implies ${(x\mapsto -)\sep r}$ and
$(x\mapsto e)\sep r$ implies $q$.
We show here how the heap update modality can be used to describe such a frame.
Let $\models \HT{p}{[x]:=e}{q}$ and
$r$ denote $\exists y([\update{x}:=y]p)$ for some fresh $y$.
By the local axiom and the frame rule, we first derive
$$
\HT{(x\mapsto -)\sep r}{[x]:=e}{(x\mapsto e)\sep r}.
$$
Let $h,s\models p$. 
To prove that
$h,s\models {(x\mapsto -)\sep r}$, it suffices to show that
there exists a split  ${h=h_1\uplus h_2}$ such 
that  $h_1,s\models {(x\mapsto -)}$ and $h_2,s[y:=n]\models [\update{x}:=y]p$,
for some $n$.
Since $\models \HT{p}{[x]:=e}{q}$
we have that $s(x)\in{\it dom}(h)$.
So we can introduce the split ${h=h_1\uplus h_2}$ such 
that  $h_1,s\models {(x\mapsto -)}$ and ${h_2=h[s(x):=\bot]}$.
By the semantics of the heap update modality it then suffices to observe that
$h_2,s[y:=h(s(x))]\models [\update{x}:=y]p$ if and only if
$h_2[s(x):=h(s(x))],s \models p$ ($y$ does not appear in $p$),
that is, $h,s\models p$.

On the other hand, we have that
$(x\mapsto e)\sep r$ implies
$q$: Let $h,s\models (x\mapsto e)\sep r$. So there exists  a  split $h=h_1\uplus h_2$ such 
that  $h_1,s\models x\mapsto e$
and $h_2,s\models r$.
Let $n$ be such that
$h_2,s[y:=n]\models [\update{x}:=y]p$.
By the semantics of the heap update modality again we have that
${h_2,s[y:=n]}\models [\update{x}:=y]p$ if and only if
$h_2[s(x):=n],s\models p$ (here $y$ does not appear in $p$).
Since $\models \HT{p}{[x]:=e}{q}$ it then follows
that ${h_2[s(x):=s(e)],s}\models q$, that is,
$h,s\models q$ (note that $h=h_2[s(x):=s(e)]$ because
$h(s(x))=s(e)$ and $h_2=h[s(x):=\bot]$).
\end{proof}

\subsection{Strongest postcondition axiomatization}
Before we discuss a novel strongest postcondition axiomatization 
using the modalities of DSL, it should be noted that in general the 
semantics of program logics which require absence of certain failures gives rise
to an asymmetry between weakest preconditions and strongest postconditions:
For any  statement $S$ and postcondition $q$ we have that
$\models \HT{\F}{S}{q}$.
However, for any precondition $p$ which does not exclude failures,
there does not exist \emph{any} postcondition $q$ such
that $\models \HT{p}{S}{q}$.
We solve this by simply requiring that the given precondition does not give rise to failures (see below).

\begin{figure}[t]
\begin{framed}
\vspace*{-6pt}
$$
{\HT{p}{x:=e}{\exists y([x:=y]p\wedge (x=e[y/x]))}}{}
$$
$$
\!\!\!\!{\HT{p\wedge (e\hookrightarrow -)}{x:=[e]}{\exists y([x:=y]p\wedge (x\hookrightarrow e[y/x]))}}{}
$$
$$
{\HT{p\wedge (x\hookrightarrow -)}{[x]:=e}{\exists y([\update{x}:=y]p)\wedge (x\hookrightarrow e)}}{}
$$
$$
\!\!{\HT{p}{x:=\new(e)}{[\dispose{x}](\exists x p)\wedge (x \hookrightarrow e)}}{}
$$
$$
\!\!{\HT{p\wedge (x\hookrightarrow -)}{\delete(x)}{\exists y([\update{x}:=y]p)\wedge (x\not\hookrightarrow -)}}{}
$$
\vspace*{-15pt}
\end{framed}
\vspace{-1em}
\caption{Strongest postcondition axioms of separation logic (SP-DSL), where $y$ is fresh everywhere and $x$ does not occur in $e$ in case of $x:=\new(e)$.}
\label{fig:SP-CSL}
\end{figure}

Figure~\ref{fig:SP-CSL} contains our novel strongest postcondition axiomatization SP-DSL, where the main novelty is in the use of the
heap update and heap clear modalities in the axiomatization of
the mutation, allocation, and de-allocation instruction.
%See the appendix, Figure~\ref{fig:hoare} for the standard proof rules of Hoare~\cite{apt2019fifty},
%and Figure~\ref{fig:SP-old} for the existing strongest postcondition axiomatization involving separating connectives \cite{Reynolds02,BannisterHK18}.
%Note that we have omitted the proof rules from Figure~\ref{fig:SP-CSL}, which are given in Figure~\ref{fig:WP}.
It is worthwhile to contrast, for example, the  use of
the heap clear modality
to express freshness in the
strongest postcondition axiomatization of the allocation instruction
with the following traditional axiom  (assuming that $x$ does not occur free
in $p$):
$$
\HT{p}{x:=\new(e)}{p\sep (x\mapsto e)}
$$
where freshness is enforced by the introduction
of the separating conjunction (which as such increases the complexity
of the postcondition).
More specifically, we have the following instance
of the allocation axiom in Figure \ref{fig:SP-CSL} (also making use of  that $x$
does not appear in the precondition)
$$
\HT{y\hookrightarrow 0}{x:=\new(1)}{
[\dispose{x}](y\hookrightarrow 0)\wedge (x \hookrightarrow 1)}
$$
Applying  \mbox{\bf E14} we obtain
$$
\HT{y\hookrightarrow 0}{x:=\new(1)}{y\not=x\wedge (y\hookrightarrow 0)
\wedge (x\hookrightarrow 1)}
$$
On the other hand, instantiating the  above traditional axiom we obtain
$$
\HT{y\hookrightarrow 0}{x:=\new(1)}{(y\hookrightarrow 0)\sep (x\mapsto 1)}
$$
which is implicit and needs
unraveling the semantics of separating conjunction.
Using the heap clear modality we thus obtain a basic
assertion in predicate logic which provides an explicit but simple
account of aliasing.

%We have the following theorem.

\begin{theorem}[Soundness and completeness SP-DSL]\label{th:SP-CSL}
For any basic instruction $S$,
we have $\models\HT{p}{S}{q}$ if and only if $\HT{p}{S}{q}$ is derivable from the axioms in SP-DSL (Figure \ref{fig:SP-CSL}) and (a single application of) the rule of consequence.
\end{theorem}

\begin{proof}
We showcase the soundness and completeness of the strongest postcondition axiomatization of
allocation 
(soundness and completeness of the strongest postconditions
for the mutation and de-allocation instructions follow in a
straightforward manner from the semantics of the heap update
modality).
\begin{itemize}
\item $\models \HT{p}{x:=\new(e)}{
[\dispose{x}](\exists y([x:=y]p))
\wedge x \hookrightarrow e}$:\\
Let $h,s\models p$. We have to show that
${h[n:=s(e)],s[x:=n]}\models {[\dispose{x}](\exists y([x:=y]p))} \wedge {x \hookrightarrow e}$,
for $n\not\in {\it dom}(h)$.
By definition $h[n:=s(e)],s[x:=n]\models x \hookrightarrow e$.
By the semantics of the heap clear modality and existential quantification,
it then suffices to show that
$h[n:=\perp],s[x:=n][y:=s(x)]\models [x:=y]p$,
which by the semantics of the simple assignment modality boils down to
$h,s[y:=s(x)]\models p$ (note that $n\not\in {\it dom}(h)$, that is,
$h,s\models p$ ($y$ does not appear in $p$), which holds by assumption.

\item $\models\HT{p}{x:=\new(e)}{q}$ implies\\
$\models ([\update{x}:=\bot](\exists y(p[x:=y]))
\wedge x \hookrightarrow e)\to q$:\\
Let 
$h,s\models [\update{x}:=\bot](\exists y([x:=y]p)) \wedge x \hookrightarrow e$.
We have to show that $h,s\models q$.
By the semantics of the heap clear modality we derive from the above assumption that
$h[s(x):=\bot],s\models \exists y(p[x:=y])$.
Let ${h[s(x):=\bot],s[y:=n]} \models p[x:=y]$, for some $n$.
It follows from the semantics of the simple assignment modality that
${h[s(x):=\bot],s[x:=n]} \models p$ ($y$ does not appear in $p$).
Since  $s(x)\not\in{\it dom}(h[s(x):=\bot])$, 
we have that 
${\langle x:=\new(e),h[s(x):=\bot],s[x:=n]\rangle} \Rightarrow {(h[s(x):=s[x:=n](e)],s)}$.
Since we can assume without loss of generality
that $x$ does not occur in $e$ 
%(as argued in Subsection \ref{sec:extensions}) 
we have that
$s[x:=n](e)=s(e)$, and so from the assumption that  $h,s\models x \hookrightarrow e$ we derive that  $h[s(x):=s[x:=n](e)]=h$.
From $\HT{p}{x:=\new(e)}{q}$ then
we conclude that $h,s\models q$.
\end{itemize}
\end{proof}

\section{Extensions}\label{sec:extensions}

%Reynold's cons (consecutive cells) and dispose. Compare to C implementation of malloc/free.
A straightforward extension concerns the general mutation instruction
$[e]:=e'$, which allows the use of an arbitrary arithmetic expression $e$ to denote the updated location.
We can simulate this by the statement $x:=e;\ [x]:=e'$, where $x$ is a fresh variable.
Applying the modalities we derive the following
axiom
$$
{\HT{(e\hookrightarrow-)\wedge [x:=e][\update{x}:=e']p}{[e]:=e'}{p}}{}
$$
where $x$ is a fresh variable.

Another straightforward extension concerns the allocation $x:=\new(e)$ in the case where $x$ does occur in $e$.
The instruction ${x:=\new(e)}$ can be simulated by ${y:=x;}\ {y:=\new(e[y/x])}$ where $y$ is a fresh variable.
Applying the sequential composition rule 
and the axiom for basic assignments,
it is straightforward 
to derive the following generalized backwards allocation axiom:
$$
{\HT{\forall y( (y\not\hookrightarrow-) \rightarrow [y:=x][\update{y}:=e[y/x]]p)}{x:=\new(e)}{p}}{}
$$
where $y$ is fresh.% and the variables of $e$ and $x$ do not occur bound in $p$.

Reynolds introduced in \cite{Reynolds02}  the allocation
instruction ${x:=\new(\bar{e})}$,
which allocates a consecutive part of the memory for
storing the values of $\bar{e}$: its semantics is described by
$$\langle x:={\bf cons}(\bar{e}),h,s\rangle \Rightarrow (h[\bar{m}:=s(\bar{e})],s[x:=m_1])$$ 
where $\bar{e}=e_1,\ldots,e_n$, $\bar{m}=m_1,\ldots,m_n$,
$m_{i+1}=m_i+1$, for $i=1,\ldots,n-1$,
$\{m_1,\ldots,m_n\}\cap \mbox{\it dom}(h)=\emptyset$,
and, finally, 
$$
h[\bar{m}:=s(\bar{e})](k)=
\left\{
\begin{array}{ll}
h(k) &\text{if }k\not\in \{m_1,\ldots,m_n\}\\
s(e_i)&\text{if }k=m_i\text{ for some }i=1,\ldots,n.
\end{array}
\right.
$$

Let $\bar{e}'$ denote a sequence of expressions
$e'_1,\ldots e'_n$ 
such that $e'_1$ denotes the variable $x$ and
$e'_{i}$ denotes the expression $x+(i-1)$, for $i=2,\ldots,n$.
The  storage of the values of $e_1,\ldots,e_n$ 
then  can be modeled by a sequence
of heap update modalities $[\update{e'_i}:=e_i]$,  for $i=1,\ldots,n$.
We abbreviate such a sequence  by
$[\update{\bar{e}'}:=\bar{e}]$.
Assuming that $x$ does not occur in one
of the expressions $\bar{e}$ (this restriction can be lifted
as described above),
we have the following generalization of the
above backwards allocation axiom
$$
{\HT{\forall x ( \Big(\bigwedge_{i=1}^n (e'_i\not\hookrightarrow-)\Big) \rightarrow [\update{\bar{e}'}:=\bar{e}]p)}{x:={\bf cons}(\bar{e})}{p}}{}
$$

\paragraph{Recursive predicates}

Next we illustrate the extension of our 
approach to recursive predicates for reasoning about a linked list.
%In general, there are two approaches in supporting recursive predicates: by means of a system of recursive predicates (see e.g. \cite{10.1007/978-3-642-38574-2_2}), or by means of an explicit
%least fixpoint binder in formulas (see e.g. \cite{10.1007/978-3-540-27815-3_36}).
%Here we follow the first approach.
Assuming a set of user-defined predicates $r(x_1,\ldots,x_n)$ of arity $n$,
we introduce corresponding basic assertions $r(e_1,\ldots,e_n)$
which are interpreted by (the least fixed point of) a system of
recursive predicate definitions
$r(x_1,\ldots,x_n) := p$, where the user-defined predicates
only occur positively in $p$.
%For a discussion of the formal semantics of such user-defined predicates we refer to the appendix \ref{app:fix}
%(this discussion, though definitely of interest, 
%is outside the scope of this paper).
%Formally the semantics of such user-defined predicates is
%defined as the least fixpoint that 
%satisfies the predicate definitions. Details are standard and
%therefore omitted.

If for any recursive definition
$r(x_1,\ldots,x_n) := p$
only the formal parameters $x_1,\ldots,x_n$ occur free in $p$,
we can simply define $[x := e]r(e_1,\ldots,e_n)$ by ${r(e_1[e/x],\ldots,e_n[e/x])}$.
However, allowing global variables in recursive
predicate definitions does affect the interpretation of these definitions.
As a very simple example, given $r(y) := x=1$,
clearly $\HT{r(y)}{x:=0}{r(y)}$ is invalid
(and so we cannot simply define  $[x:=0]r(y)$ by $r(y[0/x])$).
Furthermore, substituting the parameters of $r$ clearly does
not make sense for modalities with heap modifications (such as mutation, allocation, etc.):
as subformulas may depend on the heap, these may require alias analysis
\emph{in the definition of $r$}.

%s a redefinition
%of the recursive definitions.
%Therefore,
%we introduce the following  general method for 
%applying substitutions 
%to recursively defined predicates.
%Let $\sigma$ denote  a substitution.
%For every predicate 
%$r(x_1,\ldots,x_n)$ with defining formula $p$ we introduce a \emph{new} predicate
%$r_\sigma(x_1,\ldots,x_n)$ defined by $p\sigma$,
%where applying $\sigma$ to the basic assertion $r(t_1,\ldots,t_n)$
%yields $r_\sigma(t_1\sigma,\ldots,t_n\sigma)$.
%The corresponding substitution lemma then extends to
%recursive predicates in a straightforward manner.
%Note that applying a heap update and heap clear substitution to an expression just simplifies to the expression itself, since evaluation of expressions does not depend on the heap.

We illustrate how our dynamic logic works with recursively defined predicates
on a characteristic linked list example. In particular, let $r$ be 
the recursively defined \emph{reachability} predicate 
\[
r(x,y) := x=y \lor \exists z ((x\mapsto z) \sep  r(z,y)).
\]
We shall prove
$\HT{r(\mathit{first},y)}{\mathit{first}:=\new(\mathit{first})}{r(\mathit{first},y)}$.
To do so, we model $\mathit{first}:=\new(\mathit{first})$ by
$u:=\mathit{first};\ \mathit{first}:=\new(u)$,
for some fresh variable $u$.
Thus it is sufficient to show
\[
\HT{r(\mathit{first},y)}{u := \mathit{first};\ \mathit{first}:=\new(u)}{r(\mathit{first},y)}.
\]
%Here we thus defined reachability in terms of the separating conjunction.
%In general,  the separating connectives do \emph{not} allow for an application
%of Kleene's fixpoint theorem because they do not satisfy the requirement of Scott-continuity. Even a sub-language without  separating implication does not
%satisfy this requirement. For such a language only the Knaster–Tarski theorem
%applies and thus in general requires \emph{transfinite} induction.
%However, since
%the above recursive definition of reachability also does not involve
%occurrences of negation it allows for inductive reasoning based on a  least fixpoint semantics which is Scott-continuous.

We first calculate the weakest precondition of the last assignment: $[\mathit{first}:=\new(u)]r(\mathit{first},y)$. Using equivalence ({\bf E7}) of Lemma~\ref{lem:ABI} we obtain $\forall \mathit{first}((\mathit{first}\nothookrightarrow-)\to [\update{\mathit{first}}:=u]r(\mathit{first},y)$.

Next, we simplify the modal subformula $[\update{\mathit{first}}:=u]r(\mathit{first},y)$
we first unfold the definition of $r$, obtaining
$\mathit{first}=y \lor \exists z ((\mathit{first}\mapsto z) \sep r(z,y))$.
By Lemma~\ref{def:update} ({\bf E11}), $[\update{\mathit{first}}:=u](\mathit{first}\mapsto z \sep r(z,y))$ reduces to the disjunction of 
${(\mathit{first}\mapsto z\wedge \mathit{first}
\nothookrightarrow -) \sep 
[\update{\mathit{first}}:=u]r(z,y))}$
and ${[\update{\mathit{first}}:=u](\mathit{first}\mapsto z) \sep 
(r(z,y)\wedge \mathit{first}
\nothookrightarrow -)}$.
In the first disjunct, the left-hand side of the separating conjunction asserts that
$\mathit{first}$ is allocated (and points to $z$) and
that simultaneously $\mathit{first}$ is not allocated.
This clearly is false in every heap, so that whole disjunct
reduces to $\F$.
Simplifying the second disjunct (reducing the modality with equivalence ({\bf E10}) of Lemma~\ref{def:update}) and applying standard logical equivalences,
yields that the whole subformula is equivalent to
\[
\mathit{first}=y \lor (r(u,y)\wedge (\mathit{first}\nothookrightarrow -)).
\]
Applying the allocation axiom and an application of the consequence rule,
%(where we see that the separating conjunction can be simplified away)
we obtain
\[
\begin{array}{c}
\{\forall \mathit{first}((\mathit{first}\nothookrightarrow-)\to (\mathit{first}=y \lor  r(u,y)))\}\\
{\mathit{first}:=\new(u)}\\
\{r(\mathit{first},y)\}.
\end{array}
\]
Renaming $\mathit{first}$ by the fresh variable $f$
does not affect $r$, so
\[
\begin{array}{c}
\{\forall f ((f \nothookrightarrow-)\to (f=y \lor  r(u,y)))\}\\
{\mathit{first}:=\new(u)}\\
\{r(\mathit{first},y)\}
\end{array}
\]
can be derived.
Also substituting $u$ for $\mathit{first}$ does not affect the definition of $r$.
It then suffices to observe  that
$r(\mathit{first},y)$ (trivially) implies
$\forall f ((f \nothookrightarrow-)\to (f=y \lor  r(\mathit{first},y)))$.

%We conclude this section with the observation that in fact definitions
%of the reachability predicate implies the weak definition of the reachability predicate.

\section{Formalization in Coq}\label{sec:coq}

The main motivation behind formalizing results in a proof assistant is to %increase confidence in the correctness of its proofs.
rigorously check hand-written proofs.
%
%
%We have formalized the proofs of the
%substitution lemmas corresponding to the heap update and heap clear substitutions,
%and the basic soundness/completeness of the weakest precondition axiomatization (WP-SL)
%and the strongest postcondition axiomatizations (SP-DSL).
%For technical simplicity we restrict ourselves to the basic instructions, but it should be natural to extend the formalization of the completeness result to languages with $\mbox{\bf while}$-statements, e.g. following \cite{FAISALALAMEEN201673}.
%
%
For our formalization we used the dependently-typed calculus of inductive constructions as implemented by the Coq proof assistant.
We have used no axioms other than the axiom of function extensionality (for every two functions $f,g$ we have that $f=g$ if $f(x)=g(x)$ for all $x$). This means that we work with an underlying intuitionistic logic: we have not used the axiom of excluded middle for reasoning classically about propositions. However, the decidable propositions (propositions $P$ for which the excluded middle $P\lor \lnot P$ can be proven) allow for a limited form of classical reasoning.

We formalize the basic instructions of our programming language (assignment, look-up, mutation, allocation, and deallocation) and the semantics of basic instructions. For Boolean and arithmetic expressions we use a shallow embedding, so that those expressions can be directly given as a Coq term of the appropriate type (with a coincidence condition assumed, i.e. that values of expressions depend only on finitely many variables of the store).

There are two approaches in formalizing the semantics of assertions: shallow and deep embedding.
We have taken both approaches. In the first approach, the shallow embedding of assertions, we define assertions of DSL by their extension of satisfiability (i.e. the set of heap and store pairs in which the assertion is satisfied), that must satisfy a coincidence condition (assertions depend only on finitely many variables of the store) and a stability condition (see below). The definition of the modality operator follows from the semantics of programs, which includes basic control structures such as the {\bf while}-loop.
In the second approach, the deep embedding of assertions, assertions are modeled using an inductive type and we explicitly introduce two meta-operations on assertions that capture the heap update and heap clear modality. We have omitted the clauses for $\mbox{\bf emp}$ and $(e\mapsto e')$, since these could be defined as abbreviations, and we restrict to the basic instructions.

In the deep embedding we have no constructor corresponding to the program modality $[S]p$. Instead, two meta-operations denoted $p[\update{x} = e]$ and $p[\dispose{x}]$ are defined recursively on the structure of $p$. Crucially, we formalized and proven the following lemmas (the details are almost the same as showing the equivalences hold in the shallow embedding, Lemmas~\ref{def:update}~and~\ref{def:dispose-subst-CSL}):
\begin{lemma}[Heap update substitution lemma]\label{lem:update}\raggedright
%Let $v$ denote the value $s(e)$ of the expression $e$,
%We have
${h,s  \models p[\update{x}:=e]\text{ iff }h[s(x):=s(e)],s  \models p}$.
%where $v=s(e)$.
\end{lemma}
\begin{lemma}[Heap clear substitution lemma]\label{lem:weakdispose}\raggedright
${h,s\models p[\dispose{x}]\text{ iff }h[s(x):=\bot],s\models p}$.
\end{lemma}
By also formalizing a deep embedding, we show that the modality operator can be defined entirely on the meta-level by introducing meta-operations on formulas that are recursively defined by the structure of assertions: this captures Theorem~\ref{thm:completeness}. The shallow embedding, on the other hand, is easier to show that our approach can be readily extended to complex programs including {\bf while}-loops.

In both approaches, the semantics of assertions is classical, although we work in an intuitionistic meta-logic. We do this by employing a double negation translation, following the set-up by R.~O'Connor \cite{o2011classical}. In particular, we have that our satisfaction relation $h,s\models p$ is stable, i.e.~$\lnot\lnot(h,s\models p)$ implies $h,s\models p$. This allows us to do classical reasoning on the image of the higher-order semantics of our assertions.

The source code of our formalization is accompanied with this paper as a digital artifact (which includes the files \texttt{shallow/Language.v} and \texttt{shallow/Proof.v}, and the files \texttt{deep/Heap.v}, \texttt{deep/Language.v}, \texttt{deep/Classical.v}).
%Below we make some more detailed comments, relevant for those interested in the technical details of the formalization. 
%Readers of this section are assumed to have some basic
%knowledge of the constructive dependent type theory underlying Coq.
The artifact consists of the following files:
\begin{itemize}
\item \texttt{shallow/Language.v}: Provides a shallow embedding of Boolean expressions and arithmetic expressions, and a shallow embedding of our assertion language, as presented in the prequel.
\item \texttt{shallow/Proof.v}: Provides proof of the equivalences ({\bf E1-16}), and additionally standard equivalences for modalities involving complex programs.
\item \texttt{deep/Heap.v}: Provides an axiomatization of heaps as partial functions.
\item \texttt{deep/Language.v}: Provides a shallow embedding of Boolean expressions and arithmetic expressions, and a deep embedding of our assertion language, on which we inductively define the meta operations of heap update and heap clear. We finally formalize Hoare triples and proof systems using weakest precondition and strongest postcondition axioms for the basic instructions.
\item \texttt{deep/Classical.v}: Provides the classical semantics of assertions, and the strong partial correctness semantics of Hoare triples. Further it provides proofs of substitution lemmas 
corresponding to our meta-operators. Finally, it provides proofs of the
soundness and completeness of the aforementioned proof systems.
\end{itemize}

\section{Conclusion and related work}\label{sec:conclusion}
To the best of our knowledge no other works exist
that study dynamic logic extensions of SL.
We have shown how we can combine the standard programming logics in SL
with a new DSL axiomatization of both weakest preconditions and strongest postconditions.
These new axiomatizations in DSL have the so-called property of \emph{gracefulness}:\footnote{The term `graceful', coined by J.C. Blanchette \cite{vukmirovic2022extending}, comes from higher-order automated theorem proving where it means that a higher-order prover does not perform significantly worse on first-order problems than existing first-order provers that lack the ability to reason about higher-order problems.}
any first-order postcondition gives rise to
a first-order weakest precondition
(for any basic instruction).  A property that  existing axiomatizations of SL, such as given by C.~Bannister, P.~H\"ofner and G.~Klein~\cite{BannisterHK18}, and M.~Faisal~Al~Ameen and M.~Tatsuta~\cite{FAISALALAMEEN201673}, lack. (See also \cite{TATSUTA20191}.)
As a simple example, in our approach
$[[x]:=0](y\hookrightarrow z)$
%(which simply states that $y$ points to $z$ after setting the location $x$ to zero)
can be resolved  to the first-order formula
$$
(x\hookrightarrow-)\wedge ((y=x\wedge z=0)\vee (y\not=x\wedge y\hookrightarrow z))
$$
by applying
the above equivalences \mbox{\bf E6}
and \mbox{\bf E10}.
The standard rule for backwards reasoning in \cite{Reynolds02} 
however gives the weakest precondition:
$$
(x\mapsto -)\sep ((x\mapsto 0)\sepimp (y\hookrightarrow z))
$$
Despite their different formulations,
both formulas characterize $[[x]:=0](y\hookrightarrow z)$, and thus must be equivalent.
In fact, the equivalence has been proven
in our Coq formalization (Section~\ref{sec:coq}).
Surprisingly,
this however  exceeds the capability of all the automated SL provers in the benchmark competition for SL \cite{sighireanu2019sl}. In particular, only the CVC4-SL tool \cite{reynolds2016decision} supports the fragment of SL that includes the separating implication connective. However, from our own experiments with that tool, we found that it produces an incorrect counter-example and reported this as a bug to one of the maintainers of the project \cite{andrewcvc}. In fact, the latest version, CVC5-SL, reports the same input as `unknown', indicating that the tool is incomplete.
%The input given was:
%${(x \mapsto w)\sep ((x\mapsto 0) \sepimp ((y\mapsto z)\sep \T))}$ and\\
%$\lnot((y = x\land z = 0)\lor (y\neq x\land ((y\mapsto z) \sep \T)))$.
%TODO: IRIS.
Furthermore, we have investigated whether the equivalence of these formulas can be proven in an interactive tool for reasoning about SL: the Iris project \cite{jung2018iris}. However, also in that system it is not possible to show the equivalence of these assertions, at least not without adding additional axioms to its underlying model \cite{robbertpersonal}.
On the other hand, the equivalence between the above two
formulas can be expressed in quantifier-free separation logic, for which a complete axiomatization of all valid
formulas has been given in \cite{DemriLM21}.

In general, the calculation of $[S]p$ by means of a compositional analysis of
$p$, in contrast with the standard approach, does not generate additional \emph{nesting} of the separating connectives.
On the other hand, the compositional analysis generates a case distinction
in the definitions of $[\update{x}:=e](p\sep q)$ and
$[\dispose{x}](p\sepimp q)$.
How  the combined application of the two approaches works in practice 
needs to be further investigated.
Such an investigation will also involve the
use of the modalities for the basic instructions
in the generation of the verification conditions
of a program (as is done for example in the KeY tool \cite{DBLP:series/lncs/10001} for the verification of Java programs),
which allows to \emph{postpone} and \emph{optimize}
their actual application.
For example, the  equivalence
\[
[x:=e][\update{y}:=e']p \equiv [\update{y}:=e'[e/x]][x:=e]p
\]
allows to resolve  the simple assignment modality 
by `pushing it inside'.

Other works that investigate weakest preconditions in SL are briefly discussed below.
For example, \cite{BannisterHK18} investigates both weakest preconditions and strongest postconditions in SL,
also obtained through a transformational approach.
However, the transformation uses other  separating connectives (like \emph{septraction}), and thus is not graceful.
On the other hand, in \cite{MuraliPLM20} an alternative
logic is introduced which, instead of the separating connectives,
extends standard first-order logic with 
an operator ${\it Sp}(p)$ which captures the parts of the heap
the (first-order) formula $p$ depends on.
Thus also \cite{MuraliPLM20} goes beyond first-order, and is not graceful.
But the main motivation of that work coincides with ours:
avoiding unnecessary reasoning about the separating connectives.

Our artifact formalizes the syntax and semantics of programs and assertions of SL.
We plan to further extend our formalization  to support practical program verification,
and investigate how to integrate our approach in Iris~\cite{jung2018iris}:
we will consider how DSL
can also work for a shallow embedding of SL.
Then the generated verification conditions
require a proof of the validity of corresponding assertions in SL,
which can be discharged by providing a proof directly in Coq.
Further, we will investigate the application of DSL to concurrent SL \cite{brookes2016concurrent} and permission-based SL \cite{amighi2015permission}.

\bibliographystyle{entics}
\bibliography{refs}
\clearpage
\appendix
\section{Appendix}

\begin{figure}[ht]
\begin{framed}
$S \Coloneqq x := [e] \mid [x] := e \mid x := \new(e) \mid \delete(x) \mid $\\
\hspace*{2.5em}$x := e \mid S; S \mid \IF b \THEN S \ELSE S \FI \mid \WHILE b \DO S \OD$\\[1em]
$\langle x:=e,h,s\rangle \Rightarrow (h,s[x:=s(e)])$,\\
$\langle x:=[e],h,s\rangle \Rightarrow (h,s[x:=h(s(e))])$ if $s(e)\in\mbox{\it dom}(h)$,\\
$\langle x:=[e],h,s\rangle \Rightarrow \mbox{\bf fail}$ if $s(e)\not\in\mbox{\it dom}(h)$,\\
$\langle [x]:=e,h,s\rangle \Rightarrow (h[s(x):=s(e)],s)$ if $s(x)\in\mbox{\it dom}(h)$,\\
$\langle [x]:=e,h,s\rangle \Rightarrow \mbox{\bf fail}$ if $s(x)\not\in\mbox{\it dom}(h)$,\\
$\langle x:=\new(e),h,s\rangle \Rightarrow (h[n:=s(e)],s[x:=n])$ where $n\not\in\mbox{\it dom}(h)$.\\
$\langle \delete(x),h,s\rangle \Rightarrow (h[s(x):=\bot],s)$ if $s(x)\in\mbox{\it dom}(h)$,\\
$\langle \delete(x),h,s\rangle \Rightarrow \mbox{\bf fail}$ if $s(x)\not\in\mbox{\it dom}(h)$,\\
$\langle S_1;S_2,h,s\rangle \Rightarrow o$ if $\langle S_1,h,s\rangle \Rightarrow (h',s')$ and $\langle S_2,h',s'\rangle \Rightarrow o$,\\
$\langle S_1;S_2,h,s\rangle \Rightarrow \mbox{\bf fail}$ if $\langle S_1,h,s\rangle \Rightarrow \mbox{\bf fail}$,\\
$\langle \IF b \THEN S_1 \ELSE S_2 \FI,h,s\rangle \Rightarrow o$ if $s(b)=\T$ and $\langle S_1,h,s\rangle \Rightarrow o$,\\
$\langle \IF b \THEN S_1 \ELSE S_2 \FI,h,s\rangle \Rightarrow o$ if $s(b)=\F$ and $\langle S_2,h,s\rangle \Rightarrow o$,\\
$\langle \WHILE b \DO S \OD,h,s\rangle \Rightarrow o$ if $s(b)=\T$, $\langle S,h,s\rangle \Rightarrow (h',s')$, and $\langle \WHILE b \DO S \OD,h',s'\rangle \Rightarrow o$,\\
$\langle \WHILE b \DO S \OD,h,s\rangle \Rightarrow \mbox{\bf fail}$ if $s(b)=\T$ and $\langle S,h,s\rangle \Rightarrow \mbox{\bf fail}$,\\
$\langle \WHILE b \DO S \OD,h,s\rangle \Rightarrow (h,s)$ if $s(b)=\F$.
\end{framed}
\caption{Syntax and semantics of heap manipulating programs.}
\label{fig:full-syntax}
\end{figure}

\begin{figure}[ht]
\begin{framed}
$$
\infer{\HT{p}{S}{q}}{\models p\to p'&\HT{p'}{S}{q'}&\models q'\to q}
$$
$$
\infer{\HT{p}{S_1;S_2}{q}}{\HT{p}{S_1}{r} & \HT{r}{S_2}{q}}
$$
$$
\infer{\HT{p}{\IF b\THEN S_1\ELSE S_2\FI}{q}}{\HT{p\land b}{S_1}{q} & \HT{p\land \lnot b}{S_2}{q}}
$$
$$
\infer{\HT{p}{\WHILE b\DO S\OD}{p\land \lnot b}}{\HT{p\land b}{S}{p}}
$$
\vspace*{-15pt}
\end{framed}
\vspace{-1em}
\caption{Hoare's standard proof rules.}
\label{fig:hoare}
\end{figure}

%\begin{figure}[ht]
%\begin{framed}
%vspace*{-6pt}
%$$
%{\HT{p[x:=e]}{x:=e}{p}}{}
%$$
%$$
%{\HT{\exists y((e\hookrightarrow y) \land p[x:=y])}{x:=[e]}{p}}
%$$
%$$
%{\HT{(x\mapsto -)\sep ((x\mapsto e) \sepimp p)}{[x]:=e}{p}}{}
%$$
%$$
%{\HT{\forall y( (y\mapsto e) \sepimp p[x:=y])}{x:=\new(e)}{p}}{}
%$$
%$$
%{\HT{(x\mapsto -) \sep p}{\delete(x)}{p}}{}
%$$
%\vspace*{-15pt}
%\end{framed}
%\vspace{-1em}
%\caption{Global weakest precondition axiomatization %(cf.~\cite{Reynolds02,BannisterHK18,FAISALALAMEEN201673}).}
%\label{fig:WP-old}
%\end{figure}

\begin{figure}[ht]
\begin{framed}
\vspace*{-6pt}
$$
{\HT{p}{x:=e}{\exists y(p[x := y] \land x=e[x := y])}}{}
$$
$$
{\HT{p\wedge (e\hookrightarrow -)}{x:=[e]}{(e\mapsto x)\sep \lnot((e\mapsto x)\sepimp \lnot p)}}
$$
$$
{\HT{p\wedge (x\hookrightarrow -)}{[x]:=e}{(x\mapsto e)\sep \lnot((x\mapsto -)\sepimp\lnot p)}}{}
$$
$$
{\HT{p}{x:=\new(e)}{(x\mapsto e)\sep p}}{}
$$
$$
{\HT{p\wedge (x\hookrightarrow -)}{\delete(x)}{\lnot((x\mapsto -)\sepimp \lnot p)}}
$$
\vspace*{-15pt}
\end{framed}
\vspace{-1em}
\caption{Global strongest postcondition axiomatization (cf.~\cite{Reynolds02,BannisterHK18}),
assuming $x$ does not occur in $e$ in the axioms for look-up, mutation, and allocation.}
\label{fig:SP-old}
\end{figure}

\end{document}